\documentclass[a4paper,10pt,preprintnumbers,floatfix,superscriptaddress,aps,pra,twocolumn,showpacs,notitlepage,longbibliography,nofootinbib]{revtex4-2}

\usepackage[dvipsnames]{xcolor}
\usepackage[utf8]{inputenc}
\usepackage{amsmath}
\usepackage{amsfonts}
\usepackage{amssymb}
\usepackage{amsthm}

\usepackage{graphicx}
\usepackage{lmodern}
\usepackage{sidecap} 
\usepackage{physics}
\usepackage{mathtools}
\usepackage{dsfont}

\usepackage[colorlinks]{hyperref}
\hypersetup{
    colorlinks  = true,
    citecolor   = PineGreen,
    linkcolor   = MidnightBlue,
    urlcolor    = PineGreen
}

\graphicspath{{figures/}{./}} 

\newtheorem{thm}{Theorem}
\newtheorem{lem}{Lemma}
\newtheorem{corr}{Corollary}
\newtheorem{conj}{Conjecture}

\def\tr{\operatorname{tr}}

\def\ket#1{|#1\rangle}

\def\ketbra#1#2{|#1\rangle\langle #2 |}
\def\H#1{H\left(#1\right)}

\def\p{\mathbf{p}}
\def\q{\mathbf{q}}

\def\c2{\ensuremath{C^{(2)}}}

\newcommand{\berta}{BCCRR}
\newcommand{\zuko}{RPZ$_2$}
\newcommand{\kmu}{\text{\tiny KMU}}

\begin{document}

\title{An entropic uncertainty principle for mixed states}
\date{\today}
\author{Antonio F. Rotundo}
\affiliation{Institut f\"{u}r Theoretische Physik, Leibniz Universit\"{a}t Hannover, Germany}
\author{René Schwonnek}
\affiliation{Institut f\"{u}r Theoretische Physik, Leibniz Universit\"{a}t Hannover, Germany}
\begin{abstract}
   The entropic uncertainty principle in the form proven by Maassen and Uffink yields a fundamental inequality that is prominently used in many places all over the field of quantum information theory. In this work, we provide a family of versatile generalizations of this relation. 
   Our proof methods build on a deep connection between entropic uncertainties and  interpolation inequalities for the doubly stochastic map that links probability distributions in two measurements bases.   
   In contrast to the original relation, our generalization also incorporates the von Neumann entropy of the underlying quantum state. 
   These results can be directly used to bound the extractable randomness of a source independent QRNG in the presence of fully quantum attacks, to certify entanglement between trusted parties, or to bound the entanglement of a system with an untrusted environment. 
\end{abstract}

\maketitle

\emph{Introduction.}
Uncertainty relations express the limits imposed by quantum mechanics on our ability to either prepare a state with given properties, or measure the properties of a state to a given precision \cite{mur,entromur,busch,blw}.  
The study of uncertainty inequalities dates back to some of the most famous works of the early days of quantum theory \cite{heisenberg,schrodinger,weyl1928,robertson,kennard} and has since then stayed a topic of ongoing research \cite{hirschmann,alberto,alberto2,colesfurrer,meandist,review1,entro9,gao}.
Besides being an attractive rabbit hole by its own \cite{rudnicki2014strong,abbot,carlos,entroaddi,konrad,ourentro,rastegin,zozor,xie,wang}, having the right uncertainty relation at hand often proved to be a powerful tool \cite{berta,schneeloch,ballesterwehner,spinsqueezing,riccardi2022entanglement,ana}, e.g.\ to build a worst-case model.
For example, uncertainty relations commonly serve as an easy-to-establish estimate that allows  determining, from measured data, properties like the presence of entanglement \cite{guhne,guhne2,hofmann,variances,bergh,brunner}, or the amount of extractable secure randomness \cite{berta,tomamichel,masini,ddorf,di1,di2}.

In quantum information theory, uncertainty is typically quantified in terms of entropies. 
The prototype uncertainty relation of this type is due to  an idea of Deutsch \cite{deutsch} and  a conjecture by Kraus \cite{kraus1987complementary} proven by Maassen and Uffink \cite{muf}: 
Let $X$ and $Y$ denote two projective measurements, then the possible values of the 
Shannon entropies of their measurement outcomes, $H(X)$ and $H(Y)$, (measured on copies of a state $\rho$) are constraint by
\begin{equation}\label{eq:muff}
H(X)+H(Y)\ge S(\rho)+ c_{\kmu}\, .
\end{equation}
Here $c_{\kmu}$ (see eq.\ \eqref{eq:kmuLimit}) is a non-negative constant that depends on the overlap between the measurement bases of $X$ and $Y$, and $S(\rho)$ is the von Neumann entropy of $\rho$.  

In this note, we establish a generalization of \eqref{eq:muff} to a family of entropic uncertainty relations of the form:
\begin{align}\label{eq:goal}
\lambda H(X)+\mu H(Y) \geq \alpha S(\rho)+ c_{XY}(\alpha,\lambda,\mu)\,,
\end{align}
with parameters $\mu,\lambda,\alpha\in[0,1]$. 
More precisely, we are interested in finding a constant $c_{XY}(\alpha, \lambda,\mu)$ such that \eqref{eq:goal} holds for all states $\rho$. Our main result Thm.\ \ref{thm:bound} provides this constant by drawing a connection to the norm of the doubly stochastic map that links the probability distributions in the $X$ and the $Y$ bases.

\begin{figure}[t]
    \centering
    \includegraphics[width=0.99\linewidth]{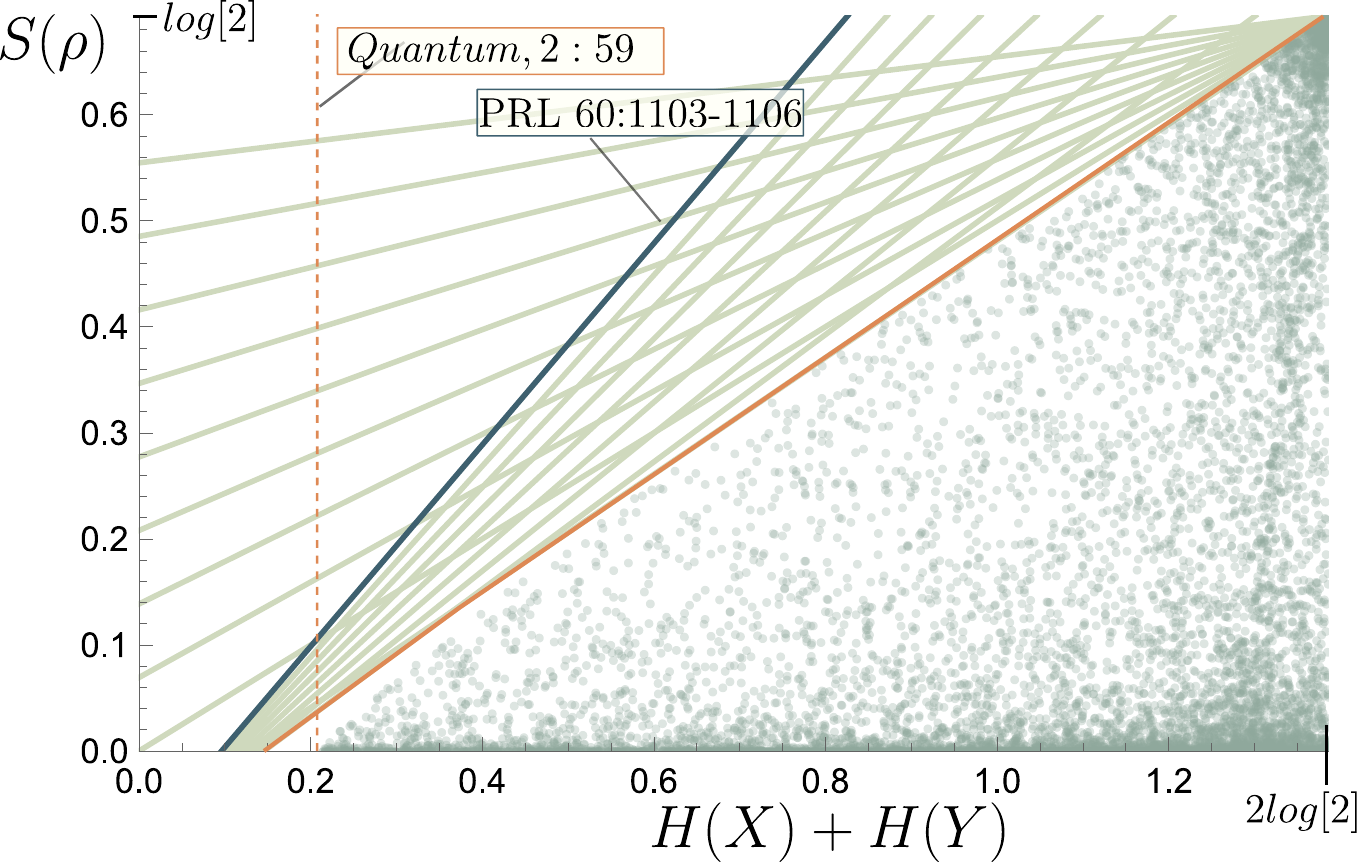}
    \caption{ 
    Allowed regions of $(S(\rho), H(X)+H(Y))$ tuples determined by random sampling over the statespace (gree dots), for $d=2$ and measurements $X$ and $Y$ with a relative angle of $17^\circ$. 
    The family of inequalities \eqref{eq:goal} determines linear bounds on this region (green lines). The uncertainty relation \eqref{eq:muff} from \cite{muf} corresponds to the blue line. Optimizing over our linear bounds (orange line) gives much stronger bounds.  
    }
    \label{fig:geometryConstraint} 
\end{figure}

The von Neumann entropy term on the r.h.s. of \eqref{eq:muff} was not present in the original formulation of this inequality. It was however noted by Frank \& Lieb \cite{hanstalk,frank} and Berta et al.\ \cite{berta}, that it can be added without changing the constant $c_{\kmu}$.  An interesting consequence, which is often overlooked, is that a state that minimizes the gap of this inequality will not necessarily be pure.  We extend this by including a weight $\alpha$ for the entropy term on the r.h.s.\ of \eqref{eq:goal}. The factor $\alpha$ sets the degree of mixedness that a state that minimizes the gap has. This goes from pure states that saturate \eqref{eq:goal} for $\alpha=0$ (see \cite{ourentro}) to the maximally mixed state that saturates \eqref{eq:goal} for $\alpha=2$ and $\mu+\lambda=1$ with $c_{XY}=0$. 
By this, we get a natural notion of a family of most certain (i.e.\ minimally uncertain) mixed states for measurements $X$
and $Y$.

An uncertainty relation like \eqref{eq:goal} can also be used to estimate the von Neumann entropy of an unknown state with given values of $H(X)$ and $H(Y)$, obtained e.g.\ from measurement data. This has various practical applications. For example, consider the scenario in which a local system $A$, with a reduced state $\rho_A$, has interacted with an uncharacterized environment $E$.  Here \eqref{eq:goal} becomes handy  for estimating correlations, since  $S(\rho_A)$ describes the corresponding entanglement entropy.
Building on this perspective, we demonstrate how to use our results for bounding the securely extractable randomness of a source independent quantum random number generator, for attesting entanglement between two trusted parties, and between two trusted parties and an uncharacterized environment.

\emph{Parameterized uncertainty relations.} Including weights, like $\alpha,\lambda,\mu$ in \eqref{eq:goal}, is a natural way of strengthening an existing relation. This has to be contrasted with many proposed \textit{improvements} of the Maassen and Uffink relation that merely add more and more $\rho$-dependent terms to the r.h.s. of \eqref{eq:muff}. 

One typical primordial question, preceding the use of an uncertainty relation, is to characterize the set of possible triples $\Omega_{XY}:=\{ (H(X),H(Y),S(\rho)) \}_\rho$ that could be attained by a not further specified state $\rho$. 
Our result \eqref{eq:goal} directly serves this purpose by giving bounds on arbitrary linear combination of $H(X)$, $H(Y)$, and $S(\rho)$.
Given a valid value of $c_{XY}$ for all parameters  $(\alpha,\mu,\lambda)$  allows for reconstructing a convex outer approximation to $\Omega_{XY}$ by performing a Legendre transformation \cite{amu}  of $c_{XY}$ with respect to $\alpha,\lambda,\mu$. 

Another typical use of an uncertainty relation is to bound the value of one quantity given access to the others. The estimation of $S(\rho)$, mentioned above, is an example of this. Here a given value of $c_{XY}$ for a whole parameter range directly pays off when we  use \eqref{eq:goal} to obtain the estimate 
\begin{equation}
    S(\rho) \leq \inf_{\lambda,\mu}  \lambda H(X) + \mu H(Y) - c_{XY}(1,\lambda,\mu).
\end{equation}
In general, this gives stronger estimates than \eqref{eq:muff}, which corresponds to evaluating the minimization above on the single point $\lambda=\mu=1$. 

The main result of this work is the following theorem, which provides a closed form for $c_{XY}$ in terms of operator norms: 
\begin{thm}\label{thm:bound}
For measurements $X$ and $Y$ given by projectors $\{X_1,\dots,X_{n_X}\}$ and $\{Y_1,\dots,Y_{n_Y}\}$, consider the $n_X\times n_Y$ matrix \c2 with entries $\c2_{ij}=\Tr(X_iY_j)$. For $0\le \lambda,\mu\le \alpha \le 1$, we have 
\begin{equation}\label{eq:thm}
c_{XY}(\alpha,\lambda,\mu)
\geq -\alpha\log\bigl\Vert \c2\bigr\Vert_{\frac{\alpha}{\mu} \rightarrow \frac{\alpha}{\alpha-\lambda}}\,.
\end{equation}
\end{thm}
The norm appearing in the theorem is defined by
\begin{equation}
\bigl\Vert \c2\bigr\Vert_{r\rightarrow s} \coloneqq \sup_{\phi\in{\mathbb{C}^{n_Y}}}
     \frac{\Vert C^{(2)}\phi\Vert_s}{\Vert\phi\Vert_r}\,,
\end{equation}
where $\norm{\cdot}_p$ denotes the usual $p$-norm. 

Plugging eq.\ \eqref{eq:thm} in eq.\ \eqref{eq:goal}, we find
\begin{equation}\label{eq:basicineq}
    \lambda H(X)+\mu H(Y) \geq \alpha S(\rho)-\alpha\log\bigl\Vert \c2\bigr\Vert_{\frac{\alpha}{\mu} \rightarrow \frac{\alpha}{\alpha-\lambda}}\,.
\end{equation}



\emph{How to evaluate the norm.}
In principle, the operator norm in the bound provided by \eqref{eq:basicineq} can be computed numerically with arbitrary precision. However, for large system sizes, this may turn out to be challenging in practise. Therefore, we provide some analytical results and a conjecture  that may drastically simplify this computation. 
From now on, we specialize to rank-1 projective measurements; in this case, we can associate orthonormal bases to $X$ and $Y$.
We denote by $d$ the number of projectors, i.e. the dimension of the Hilbert space. 

When the bases of $X$ and $Y$ are mutually unbiased (MUB), \c2 has all entries equal to $1/d$. Its norm is (Lemma 3 in \cite{additional})
\begin{equation}\label{eq:normMUB}
    \log\bigl\Vert\c2_{\text{\tiny MUB}}\bigr\Vert_{\frac{1}{\mu}\rightarrow \frac{1}{1-\lambda}} =(1-\lambda-\mu)\log{d}\,.
\end{equation}
One can show (Lemma \ref{lem:normExact} in \cite{additional}), that this equation in fact holds for any matrix $\c2$ as long as $\mu+\lambda\le1$.
In the limit $\mu=\lambda\rightarrow 1$, we recover the well-known result of Kraus, Maassen and Uffink (KMU) \cite{kraus1987complementary, muf}, 
\begin{equation}\label{eq:kmuLimit}
    \bigl\Vert\c2\bigr\Vert_{1\rightarrow \infty}= \max_{ij}\c2_{ij}\,.
\end{equation}

What we have learned about these norms is summarized in Fig.\ \ref{fig:normSummary}. 
\begin{figure}
    \includegraphics[width=0.99\linewidth]{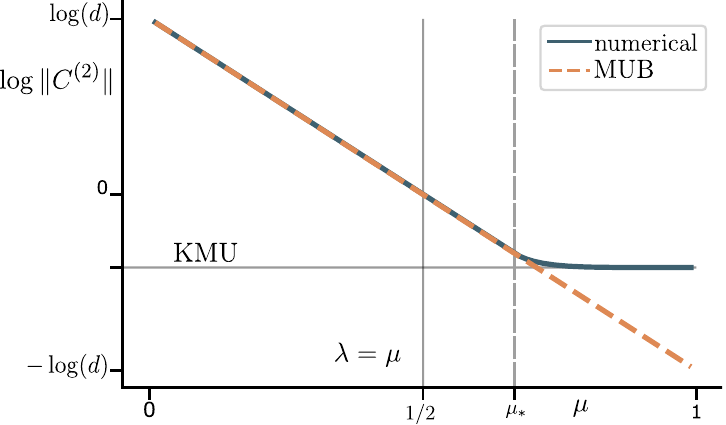}
    \caption{Typical behaviour of $\log\Vert\c2\Vert$ as a function of $\mu$, for $\mu=\lambda$ (blue). It follows eq.\ \eqref{eq:normMUB} (orange) up to a critical value $\mu_*$, and then asymptotes to the KMU limit, eq.\ \eqref{eq:kmuLimit}. To generate the plot, we consider an example \c2 in $d=2$ (eq.\ \eqref{eq:2dDoublyStoc} with $\theta = \pi/6$), but a similar behaviour is observed in higher dimensions. 
    }
    \label{fig:normSummary}
\end{figure}
We observe an interesting fact: the norm of \c2 seems to follow the MUB result beyond $\mu=1/2$, up to a critical value $\mu_*$. One can check that a similar behaviour is also found in higher dimensions, and for $\mu\neq \lambda$. This leads us to formulate the following conjecture. 
\begin{conj}\label{conjecture}(Extended MUB regime).
Let \c2 be a doubly stochastic matrix. Its norm is equal to that of $\c2_{\text{\tiny MUB}}$, i.e.\ it is given by eq.\ \eqref{eq:normMUB}, as long as 
\begin{equation}\label{eq:conjecture}
    \frac{1-\mu}{\mu}\frac{1-\lambda}{\lambda}\ge\sigma_2^2\,,
\end{equation}
where $\sigma_2$ is the second largest singular value of \c2.
\end{conj}
\noindent For $\mu=\lambda$ we can simplify the condition above to 
\begin{equation}\label{eq:qudit_mustar}
    \mu \le \frac{1}{1+\sigma_2}\equiv\mu_*\,.
\end{equation}
 Strong evidence supporting this conjecture is given in \cite{additional}. It is easy to confirm this conjecture for qubits and check it pointwise, i.e.\ for fixed parameters, in higher dimensions. Below, we will consider different applications of eq.\ \eqref{eq:basicineq}, and use the analytical expression given by conjecture \ref{conjecture} to find optimal values for $\mu$ and $\lambda$. For critical applications, such as security proofs, one can then easily check the validity of the analytical expression, for that specific point, by numerically evaluating the norm. 
 






\emph{Comparison to existing uncertainty relations.}
In this section, we compare \eqref{eq:basicineq} with two other known EUR: that of \cite{berta}, which we denote \berta,
\begin{equation}
    H(X)+H(Y)\ge S(\rho) -\log c_1\,;
\end{equation}
and the inequality \zuko\ from \cite{rudnicki2014strong},
\begin{equation}
    H(X)+H(Y)\ge S(\rho) -\log[c_1C^2+c_2(1-C^2)]\,.
\end{equation}
Here, following \cite{rudnicki2014strong}, $c_1$ and $c_2$ denote the first and second largest elements of \c2, and $C\equiv(1+\sqrt{c_1})/2$. 

To compare with these inequalities, which give equal weights to $H(X)$ and $H(Y)$, we set $\mu=\lambda$. Using conjecture \ref{conjecture} and setting $\mu=\mu_*$ as in eq.\ \eqref{eq:qudit_mustar}, we find
\begin{equation}\label{eq:quditEUR}
    H(X)+H(Y)\ge (1+\sigma_2)S(\rho) +(1-\sigma_2)\log{d}\,.
\end{equation}
The entropy term in eq.\ \eqref{eq:quditEUR} is always larger than in both \berta\ and \zuko. Therefore, in our comparison, we consider only the second, state-independent, term. 


We first consider $d=2$. In this case, the most general \c2 matrix is given by 
\begin{equation}\label{eq:2dDoublyStoc}
    \c2 = \begin{pmatrix}
        \cos^2{\theta}& \sin^2{\theta}\\
        \sin^2{\theta} & \cos^2{\theta}
    \end{pmatrix}\,, \quad\theta\in[0,\pi/4]\,,
\end{equation}
where $\theta=0$ corresponds to $X=Y$, and $\theta=\pi/4$ to the MUB case. 
\begin{figure}
\includegraphics[width=0.95\linewidth]{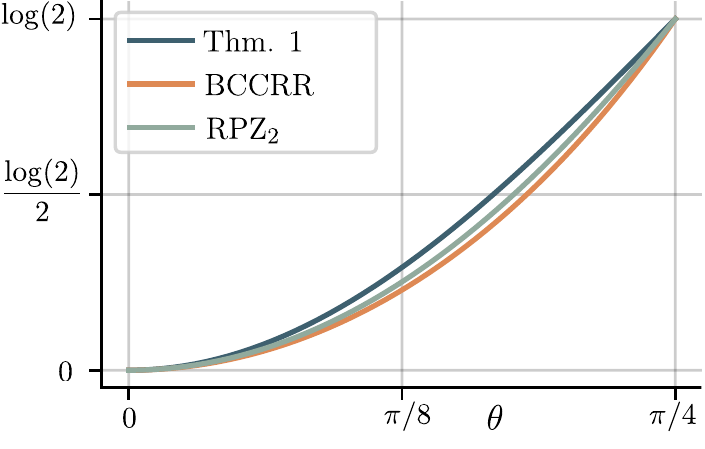}
    \caption{Comparison of the state independent bound provided by \berta\, \zuko, and eq.\ \eqref{eq:quditEUR} for $d=2$. The angle $\theta$ parametrizes \c2, as in eq.\ \eqref{eq:2dDoublyStoc}.}
    \label{fig:qubitComp}
\end{figure}
The bounds provided by \berta, \zuko, and eq.\ \eqref{eq:quditEUR} are compared in Fig.\ \ref{fig:qubitComp}.
The bounds are equivalent for $\theta=0$ and $\theta=\pi/4$, while in between our bound is stronger.

For $d>2$, the matrices \c2 have too many parameters, and we cannot scan them all, as we did in Fig.\ \ref{fig:qubitComp}. Instead, we compare our bound, eq.\ \eqref{eq:quditEUR}, to \berta\ and \zuko\ for some random \c2, generated by setting $\c2_{ij}=\abs{U_{ij}}^2$, where $U$ is a Haar random unitary. 
In Fig.\ \ref{fig:bad_c2_d}, we plot the percentage of \c2 for which our bound is better than \berta\ or \zuko\ as a function of $d$. 
\begin{figure}
    \centering
    \includegraphics[width = 0.9\linewidth]{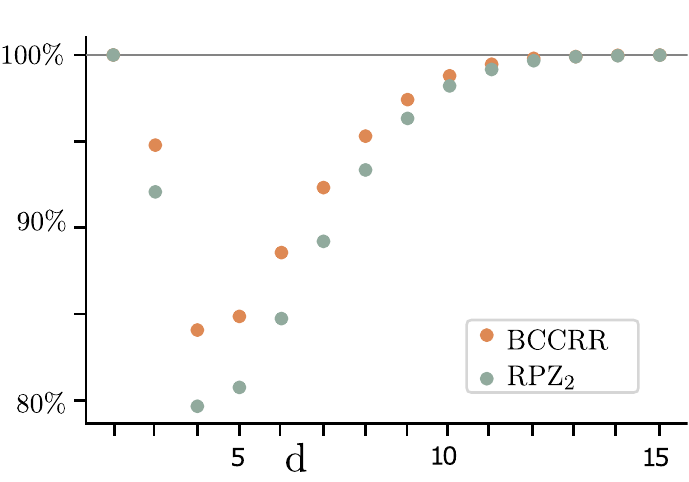}
    \caption{Percentage of \c2 for which our bound is better than BCCRR or RPZ, as a function of $d$. We have used a sample of $10^5$ random \c2.}
    \label{fig:bad_c2_d}
\end{figure}
As expected, for $d=2$, our bound is always at least as good as both \berta\ and \zuko. The percentage of \c2 for which our bound is better decreases for $d=3,4$, but then starts increasing. For $d\gg 1$, our bound is equivalent or better with probability that approaches 1.  



\emph{Practical applications.}
From the comparison above, we know that eq.\ \eqref{eq:basicineq}, for $\mu=\lambda$,  provides stronger constraints compared to other EUR for many observables $X$ and $Y$ (eq.\ \eqref{eq:quditEUR}). Therefore, it can be useful in all applications of other EUR. However, imposing $\mu=\lambda$ we loose part of the power of eq.\ \eqref{eq:basicineq}, i.e.\ that of giving different weights to $H(X)$ and $H(Y)$. This is often beneficial in practical applications, as we show  
in three examples: bounding extractable randomness, entanglement detection, and bounding entanglement with an eavesdropper. 

\emph{Bounding extractable randomness.}
Quantum random number generators will likely be one of the first competing market-ready quantum devices. However, proving their security without imposing too strong assumptions is still under development. A promising class of protocols are source-independent random number generators. Their basic security mechanism can be traced back to the use of the uncertainty relation \eqref{eq:muff}. Our results can be directly used to get stronger bounds on the extractable randomness. 

In a basic protocol \cite{cao}, we are provided with a state $\rho$, emitted by an untrusted source, from which we want to extract random numbers. 
We are allowed to perform measurements $X$ and $Y$. By convention, the $X$ measurement will be used for generating a secret number. The entropy $H(Y)$ of the other measurement, in this context usually referred to as phase-error rate, will be used to certify properties of $\rho$. We consider fully quantum attacks, modeled by granting an adversary $E$ full access to the purification of $\rho$. 

It was shown in \cite{rotem,rennerphd} that the single-shot quantity that has to be bounded for estimating the rate of securely extractable randomness (both asymptotically and finite) is given by the conditional entropy $H(X\vert E)$, which in our case can equivalently \cite{tan} be computed by $H(X)-S(\rho)$. 
Using \eqref{eq:basicineq}, we can bound this expression as 
\begin{equation}\label{eq:randomness}
    H(X\vert E)\ge \max_{\mu,\lambda}\,(1-\lambda)H(X)-\mu H(Y)-\log\Vert\c2\Vert\,.
\end{equation}
The main advantage of using eq.\ \eqref{eq:basicineq} is that we can optimize over $\mu$ and $\lambda$ to obtain improved bounds compared to other symmetric EUR. The optimal value can be found by numerically evaluating the norm. 

Using conjecture \ref{conjecture}, we can avoid numerics and get an analytical bound. Let $\Delta_X\equiv \log{d}-H(X)$, $\Delta_Y\equiv\log{d}-H(Y)$, then one can show that, as long as \eqref{eq:conjecture} is satisfied, the optimal bound is \cite{additional}
\begin{equation}\label{eq:randomConj}
    H(X\vert E)\ge \begin{dcases}
        \frac{(\sqrt{\Delta_Y}-\sigma_2\sqrt{\Delta_X})^2}{1-\sigma_2^2}&\quad \gamma\ge \sigma_2\,,\\
        \Delta_Y - \Delta_X&\quad \gamma<\sigma_2\,,
    \end{dcases}
\end{equation}
where $\gamma\equiv\sqrt{\Delta_X/\Delta_Y}$. Notice that wlog we can assume that $\gamma\le 1$ (swap $X$ and $Y$ if this is not the case). 

\emph{Entanglement detection.}
Consider now a bipartite state $\rho_{AB}$, shared between two parties, $A$ and $B$, that can perform local measurements, $X_{AB}=X_A\otimes X_B$, $Y_{AB}=Y_A\otimes Y_B$, and exchange classical information. 
To use eq.\ \eqref{eq:basicineq} for detecting entanglement, we follow \cite{riccardi2022entanglement}; one can show that $\rho_{AB}$ is entangled if \cite{additional}
\begin{equation}\label{eq:entDet2}
    \lambda H(X_{AB})+\mu H(Y_{AB}) < S_{max}-\log\bigl\Vert\c2_A\bigr\Vert\bigl\Vert\c2_B\bigr\Vert\,.
\end{equation}
Here $S_{max}=\max(S(\rho_A),S(\rho_B))$. 

We can use conjecture \ref{conjecture} to find a condition that is easier to treat analytically. Let $\sigma_2=\max(\sigma_{2,A}, \sigma_{2,B})$, where $\sigma_{2,A}$ and $\sigma_{2,B}$ are the second largest singular values of $\c2_A$ and $\c2_B$. Then, as long as $\mu, \lambda$ obey eq.\ \eqref{eq:conjecture}, we find that $\rho_{AB}$ is entangled if
\begin{equation}\label{eq:entWitnessMUB}
    \lambda \Delta_X+\mu \Delta_Y>\log{d}-S_{max}\,.
\end{equation}
Here, $\Delta_X\equiv \log{d}-H(X_{AB})$, $\Delta_Y\equiv\log{d}-H(Y_{AB})$, and $d=d_Ad_B$ is the total size of the Hilbert space. 
Since both $\Delta_X$ and $\Delta_Y$ are non-negative, the best we can do is to take $\mu$ and $\lambda$ as big as possible. However, the constraints \eqref{eq:conjecture} prevents us from increasing $\mu$ and $\lambda$ independently.
The optimal value of $(\mu,\lambda)$ for given $\Delta_{X,Y}$ is \cite{additional}
\begin{equation}\label{eq:optimalMuLambda}
    \mu=\frac{1-\sigma_2\gamma}{1-\sigma_2^2}\,, \quad \lambda=\frac{1-\sigma_2/\gamma}{1-\sigma_2^2}\,,\quad \gamma\equiv\sqrt{\frac{\Delta_X}{\Delta_Y}}\,,
\end{equation}
if $\sigma_2\le\gamma\le1/\sigma_2$. When $\gamma  <\sigma_2$, the optimal choice is $(\mu,\lambda)=(1,0)$, and, when $\gamma  >1/\sigma_2$, it is $(\mu,\lambda)=(0,1)$. 
In the original notation, in terms of $H(X)$ and $H(Y)$, we find that, for $\sigma_2\le\gamma\le 1/\sigma_2$, $\rho_{AB}$ is entangled if
\begin{equation}\label{eq:entWitness}
\begin{split}
    H(X_{AB})+H(Y_{AB})<&(1-\sigma_2^2)S_{max}+(1+\sigma_2^2)\\
    &\cdot\log{d}-2\sigma_2\sqrt{\Delta_X\Delta_Y}\,.
\end{split}
\end{equation}
We conclude that the freedom of keeping $\mu\neq\lambda$ generally helps. However, notice that for MUB $\sigma_2=0$, and we can take $\mu=\lambda=1$. 

As an example, we consider Werner states \cite{werner1989quantum}
\begin{equation}
W = \frac{1}{d(d^2-1)}[(d-\Phi)\mathds{1}+(d\Phi-1)V],
\end{equation}
where $V=\sum_{ij}\vert ij\rangle\langle ji\vert$ is the swap operator, and $\Phi\in[-1,1]$. These states are entangled for $\Phi<0$. We consider Werner states of two qubits, $d_A=d_B=2$, take $X_{A,B}$ in the computational basis, and $Y_{A}$,$Y_{B}$ in bases rotated by angles $\theta_{A}$, $\theta_{B}$. 
In Fig.\ \ref{fig:werner}, we plot for which values of $(\theta_A, \theta_B)$ eq.\ \eqref{eq:entWitness} can detect entangled Werner states for various values of $\Phi<0$.  
\begin{figure}
    \includegraphics[width=0.95\linewidth]{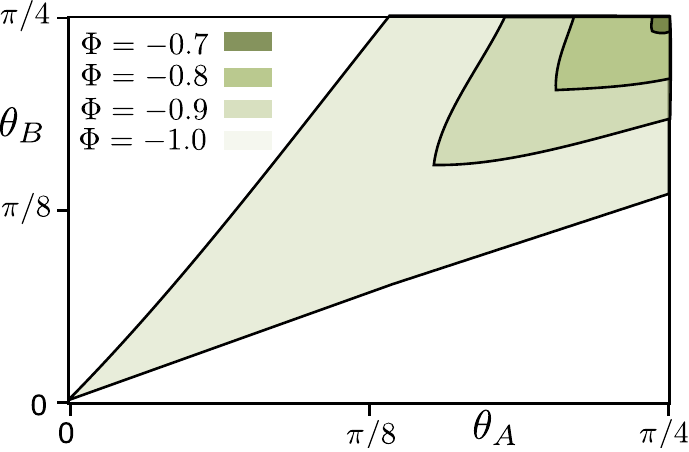}
    \caption{Values of $\theta_A$ and $\theta_B$ for which entangled Werner states are detected for various values of $\Phi$. 
    }
    \label{fig:werner}
\end{figure}
The set of measurements able to detect entangled Werner states shrinks as we increase $\Phi$. Measurements in MUB (top-right corner) are the most powerful for detecting entanglement. 

\emph{Bound on entropy.}
Consider the same setup as in the entanglement detection example; A and B are now interested in quantifying how much entanglement they might share with an eavesdropper, E. Let $\rho_{ABE}$ be the joint state of A, B, and E; in the worst case, this state is pure, and $S(\rho_E)=S(\rho_{AB})$. A direct application of eq.\ \eqref{eq:basicineq} gives the following bound:
\begin{equation}
    S(E)\le \min_{\mu, \lambda}\lambda H(X_{AB})+\mu H(Y_{AB}) + \log\Vert\c2_{AB}\Vert\,. 
\end{equation}

To obtain an analytic result, we can again use our conjecture, and proceed similarly to what we did above for entanglement detection. 
As long as $\sigma_2\le\gamma\le 1/\sigma_2$, we find that the optimal choices for $\mu$ and $\lambda$ are again given by \eqref{eq:optimalMuLambda}, where $\sigma_2$ is now the second largest singular value of $\c2_{AB}$. We arrive at the following improved bound for $S(E)$,
\begin{equation}\label{eq:entBound}
\begin{split}
    S(E)\le \frac{1}{1-\sigma_2^2}\bigl[&H(X)+H(Y)+2\sqrt{\Delta_X\Delta_Y}\sigma_2\\
    &-(1+\sigma_2^2)\log{d_Ad_B}\bigr]\,.
\end{split}
\end{equation}
Notice that, for MUB, we find again that the best we can do is to set $\mu=\lambda=1$.


\emph{Conclusions.}
In this article, we have introduced a new class of EUR, eq.\ \eqref{eq:basicineq}, which allows giving different weights to $H(X)$, $H(Y)$, and $S(\rho)$. We have shown that these EUR often provide better bounds than other EUR known in the literature. Moreover, we have shown in three examples that the freedom of giving different weights can be helpful in practical applications. Eq.\ \eqref{eq:basicineq} is expressed in terms of norms, which are difficult to estimate numerically. We have explored properties of these norms, and formulated a conjecture that, if correct, leads to an analytic result valid for most of the parameter space. In particular, we have obtained eq.\ \eqref{eq:quditEUR}, that provides a simple alternative to Maassen-Uffink, where $\log{c_{MU}}$ is replaced by $(1-\sigma_2)\log d$. 

There are several directions in which  this work could be extended. Clearly, it would be nice to prove conjecture \ref{conjecture}. Also,
using conjecture \ref{conjecture}, we can analytically study eq.\ \eqref{eq:basicineq} for values of $(\mu,\lambda)$ satisfying eq.\ \eqref{eq:conjecture}. Close to $\mu=\lambda=1$, the KMU result applies. It remains to explore eq.\ \eqref{eq:basicineq} for intermediate values of $(\mu, \lambda)$. 
Finally, for applications to QKD, it would be interesting to extend eq.\ \eqref{eq:basicineq} to Renyi and conditional entropies. 

\begin{acknowledgments}
We gratefully acknowledge discussions with Mario Berta, Giovanni Chesi, Thomas Cope, Andreea Lefterovici, Lorenzo Maccone, Tobias Osborne, Marco Tomamichel, and Henrik Wilming. AR thanks the Quantum Information Theory Group (QUit) of the University of Pavia for hospitality during the completion of this work. AR acknowledges financial support by the BMBF project QuBRA. RS acknowledges financial support by the Quantum Valley Lower Saxony and by the BMBF project ATIQ. 
\end{acknowledgments}

\bibliographystyle{unsrt}
\bibliography{entrolib}

\newpage
\appendix
\onecolumngrid\section{Appendix}

\subsection{Proof of the main theorem}
\begin{proof}
\begin{align}
c(\alpha,\lambda,\mu)=\inf_{\rho} \quad \lambda H(X)+\mu H(Y) - \alpha H(\rho) \label{whatwewant}
\end{align}
Before we start the proof, we have to introduce some notation: essential steps of this proof rely on maximizing functionals on unit balls according to vector p-norms. We will use the notation
\begin{align}
B^p(\Omega)=\{\mathbf v \in \Omega| \Vert \mathbf v \Vert_p\leq 1 \}
\end{align}
to denote the unit norm ball restricted to a set $\Omega$. Within this notation the set of all probability distributions on a set with $d$ elements is denoted by $B^1(\mathbb{R}_+^d)$, which we will also abbreviate by $P_d$ for convenience. In the same manner we will abbreviate the set of all quantum states on a $n$-dimensional Hilbertspace by $S_n$   

(i) \textit{Linearization of entropies -} We start the proof by rewriting all entropies in \eqref{whatwewant} as variation of linear functionals: 
The relative entropy $D(\p|\q)=\sum p_i ( \log(p_i)-\log(q_i))$ between any pair of probability distributions $\p,\q\in P_d$ is positive and zero if and only if $\p=\q$. Hence, we can conclude, 
\begin{align}
0=\inf_{\q\in P_d} D(\p|\q)=\inf_{\q\in P_d} \sum p_i \left(\log(p_i)-\log(q_i)\right)=\inf_{\q\in P_d}  -\sum p_i \log(q_i)-H(\p),
\end{align}
for any fixed $\p$. This directly allow us to characterize the Shannon entropy by the variation 
\begin{align}
\H \p = \inf_{\q\in P_d}  -\sum p_i \log(q_i).\label{shannonvar}
\end{align}
In the same manner we can use the positivity of the quantum relative entropy $D(\rho|\sigma)$, between states $\rho,\sigma\in S_n$, to conclude an analogous characterization of the von Neumann entropy
\begin{align}
\H \rho= \inf_{\sigma\in S_n} -\tr\left(\rho \log \sigma \right).\label{neumannvar}
\end{align}
The entropy $\H X$ of a measurement $X$ is computed as the Shannon entropy of its output distribution, i.e. as entropy of the distribution $\p^\rho$ with entries $p^\rho_i=\tr \left(\rho X_i\right)$. Using the variational characterization \eqref{shannonvar} we can rewrite this as 
\begin{align}
\H X =- \inf_{\p\in P_d} \sum \tr(\rho X_i) \log(p_i)=-\inf_{\p\in P_d} \tr\left( \rho \sum_{i=1}^{n_X} X_i \log(p_i)\right).
\end{align}
Introducing the operators 
\begin{align}
A_X(\p):=-\sum_{i=1}^{n_X} X_i \log(p_i) \text{ \quad and \quad } A_Y(\q):=-\sum_{j=1}^{n_Y} Y_j \log(q_j)
\end{align}
for the measurements $X$ and $Y$ respectively, then enables us to rewrite the Shannon entropy of our measurement outcomes as
\begin{align}
\H{X} = \inf_{\p\in B_1^+} \tr\left( \rho A_X(\p)\right) \text{ and } \H{Y} = \inf_{\q\in B_1^+} \tr\left( \rho A_Y(\q)\right).  \label{shannonvarXY}
\end{align}
Together with \eqref{neumannvar} we can rewrite the l.h.s. of \eqref{whatwewant}, our desired quantity,  as 
\begin{align}
c(\alpha,\lambda,\mu)=\inf_{\rho\in S_n} \sup_{\sigma\in S_n} \inf_{\p,\q\in P_d}
\tr\left(\rho
 \left(
 	\lambda A_X(\p)+ \mu A_Y(\q) +\alpha \log(\sigma) 
 \right)
\right).
\end{align}
Exchanging an $\inf \sup$ with a $\sup \inf$ will then give us a lower bound
\begin{align}
c(\alpha,\lambda,\mu)\geq\sup_{\sigma\in S_n} \inf_{\rho\in S_n} \inf_{\p,\q\in P_d}
\tr\left(\rho
\left(
\lambda A_X(\p)+ \mu A_Y(\q) +\alpha \log(\sigma) 
\right)
\right).\label{whatwewanttobound}
\end{align}

(ii)  \textit{Representation as operator norm:}
For a self adjoint operator, say $K$, the minimization of a functional $\tr(\rho K)$  over all $\rho$ will be attained on an eigenstate corresponding to an extremal (the smallest) eigenvalue. This extremal eigenstate will still be an optimizer if we apply a monotone decreasing function, say $f$, to $K$ and consider a maximization of $\tr(\rho f(K))$ instead. Following this idea we can conclude the identity 
\begin{align}
\inf_{\rho\in S_n} \tr\left( \rho K \right)=-\log\left(\sup_{\rho \in S_n}\tr(\rho e^{-K})  \right),\label{normexp}
\end{align}
where we applied the mon. decreasing function $f=\exp(-x)$ first and its inverse $f^{-1}=-\log(x)$ afterwards.
Using \eqref{normexp} on \eqref{whatwewanttobound} then gives
\begin{align}
c(\alpha,\lambda,\mu)&\geq \sup_{\sigma\in S_n} \inf_{p,q\in P_d} -\log\left(\sup_{\rho\in S_n}
\tr\left(\rho
e^{-\lambda A_X(\p)-\mu A_Y(\q) - \alpha \log(\sigma) }
\right)
\right)\nonumber\\
&= \sup_{\sigma\in S_n} \inf_{\p,\q\in P_d} -\log
\left\Vert
e^{-\lambda A_X(\p)-\mu A_Y(\q) - \alpha \log(\sigma) }
\right\Vert_\infty.
\end{align}
where we parsed the maximization over $\rho$ as operator norm.

(iii) \textit{Fixing $\sigma$:}
We can bound the above from below by considering a fixed $\sigma$. For $\alpha\neq 0$, choosing 
\begin{align}
\sigma=\frac{e^{-1/\alpha(\lambda A_X(\p)+\mu A_Y(\q))}}
{N}
\end{align}
with a normalization constant 
\begin{align}
N=\tr\left( e^{-1/\alpha(\lambda A_X(\p)+\mu A_Y(\q)) }\right)
\end{align}
gives 
\begin{align}
c(\alpha,\lambda,\mu)&\geq \inf_{\p,\q\in P_d} -\log
\left\Vert
e^{0 -\alpha \log(1/N) \mathbb{I}}
\right\Vert_\infty
=\inf_{\p,\q\in P_d} -\alpha\log\left(
N\right)\nonumber\\
&= \inf_{\p,\q\in P_d} -\alpha\log\left(
\tr\left(
e^{- \lambda/\alpha A_X(\p)-\mu/\alpha  A_Y(\q) }\right)\right)\nonumber\\
&= -\alpha\log\left(\sup_{\p,\q\in P_d}
\tr\left(
e^{- \lambda/\alpha A_X(\p)-\mu/\alpha  A_Y(\q) }\right)\right).\label{gleichkommtgt}
\end{align}
(iv)\textit{ Golden Thompson:} Since the function $-\log$ is mon. decreasing, we can get lower bounds on $c(\lambda,\mu,\alpha)$ by upper bounding the argument of the logarithm in the above. The structure of \eqref{gleichkommtgt} suggests to use the Golden Thompson inequality \cite{golden} \cite{thompson} 
\begin{align}
\tr\left(e^{A+B}\right) \leq \tr\left( e^{A} e^B\right)
\end{align}
in order to get the bound 
\begin{align}
c(\lambda,\mu,\alpha)&\geq
-\alpha\log\left(\sup_{\p,\q\in P_d}
\tr \left(
e^{- \lambda/\alpha A_X(\p)} e^{-\mu/\alpha  A_Y(\q) }\right)\right)\nonumber\\
&= 
-\alpha\log\left(\sup_{\p,\q\in P_d}
\tr \left(
\Bigl(\sum_{i=1}^{n_X}e^{\lambda/\alpha \log(p_i)}X_i\Bigr)\Bigl(
\sum_{j=1}^{n_Y}e^{\lambda/\alpha \log(q_j)}Y_j\Bigr) \right)\right)\nonumber\\
&= 
-\alpha\log\left(\sup_{\p,\q\in P_d}
\sum_{i=1,j=1}^{n_X,n_Y} p_i^{\lambda/\alpha}  q_j^{\mu/\alpha}
\tr\left(X_i Y_j\right) \right)
\end{align}

In the above we can substitute the matrix $C^{(2)}$ with entries 
$C^{(2)}_{ij}=\tr\left(X_i Y_j\right) $ and observe that, for all $p,q\in\mathcal{P}_d$, 
\begin{align}
   \left( p_1^{\lambda/\alpha} ,\dots  ,p_{n_X}^{\lambda/\alpha} \right) \in \mathcal{B}^{\alpha/\lambda} (\mathbb R_+) \text{ and } \left( q_1^{\lambda/\alpha} ,\dots  ,q_{n_Y}^{\lambda/\alpha} \right) \in \mathcal{B}^{\alpha/\mu} (\mathbb R_+). 
\end{align}
We arrive at 
\begin{align}
    c(\lambda,\mu,\alpha)&\geq -\alpha\log\left(
\sup_{\psi\in \mathcal{B}^{\alpha/\lambda}, \phi  \in \mathcal{B}^{\alpha/\mu} }
 \langle \psi |C^{(2)}|\phi \rangle\right)
\end{align}

By applying Hölder's inequality we get 
\begin{align}
\sup_{\psi\in \mathcal{B}^{\alpha/\lambda}, \phi  \in \mathcal{B}^{\alpha/\mu} }
 \langle \psi |C^{(2)}|\phi \rangle
 = \sup_{\phi\in\mathbb C^d} 
 \frac{ \Vert C^{(2)} \phi \Vert_{\frac{\alpha}{\alpha-\lambda}} }
 {\Vert  \phi \Vert_{\alpha/\mu} },
\end{align}
which gives the desired bound in Thm.\ref{thm:bound}
\end{proof}

\begin{corr}(Gibbs variational principle)\label{corrlieb}
Let $L$ be a self adjoined operator. For all quantum states $\rho$ the estimate 
\begin{align}
\tr( \rho L) \geq H(\rho) - \log\tr\left(e^{-L}\right)
\end{align}
holds.
\end{corr}
\begin{proof}
Consider the thermal state $\sigma=e^A/\tr(e^A)$. We have 
\begin{align}
\tr(\rho L) 
&=
\tr\left(\rho \log e^{-L}\right)= \tr\left(\rho \log \frac{e^{-L}}{\tr(e^{-L})}\right)- \log\tr\left(e^{-L}\right)=-\tr\left(\rho \log \sigma \right)-\log\tr\left(e^{-L}\right)\nonumber\\
&\geq 
H(\rho)- \log\tr\left(e^{-L}\right),
\end{align}
which holds since $D(\rho||\sigma)=\tr( \rho \log \rho -\rho\log \sigma)\geq 0$
\end{proof}

\section{Properties of $\norm*{\c2}_{r\rightarrow s}$}\label{app:normProp}
In this appendix, we collect results on norms
\begin{equation}\label{eq:normDef}
    \lVert\c2\rVert_{r\rightarrow s}\coloneqq \sup_{\phi \in \mathbb{C}^d}\frac{\lVert\c2 \phi\rVert_s}{\lVert \phi\rVert_r}\,,
\end{equation}
where \c2 is a $d\times d$ doubly stochastic matrix, obtained from a unitary matrix $U$, according to $\c2_{ij}=\abs{U_{ij}}^2$. The norm $\norm{\cdot}_{r\rightarrow s}$ is the operator norm between $\mathbb{C}^{d}$ equipped with the $p$-norm, for $p=r$, and again $\mathbb{C}^{d}$ equipped with the $p$-norm, with $p=s$. The $p$-norm is a well defined norm only for $p\ge1$; therefore, we only consider $r,s\ge 1$.

\begin{lem}\label{lem:posVectors}(Optimize over positive vectors). 
In general, to calculate $r\rightarrow s$ norms we need to optimize over complex vectors $\phi\in \mathbb{C}^{d}$. To calculate the norm of \c2, with $\c2_{ij}=\abs{U_{ij}}^2$, it is enough to optimize over positive vectors: $x=(x_1, x_2, \dots, x_d)\in\mathbb{R}_+^d$, with $x_i\ge 0$,
\begin{equation}
    \lVert\c2\rVert_{r\rightarrow s}= \sup_{x \in \mathbb{R}_+^d}\frac{\lVert\c2 x\rVert_s}{\lVert x\rVert_r}\,.
\end{equation}
\end{lem}
\begin{proof}
This follows from the definition of the norm and the special form of \c2,
\begin{equation}
\begin{split}
    \lVert\c2\rVert_{r\rightarrow s}&= \sup_{\phi \in \mathbb{C}^d}\frac{\lVert\c2 \phi\rVert_s}{\lVert \phi\rVert_r}\\
    &=\sup_{x\in \mathbb{R}_+^d}\sup_{\theta \in [0, 2\pi]^d}\frac{\bigl[\sum_i\bigl(\sum_j \abs{U_{ij}}^4x_j^2+2\sum_{j<j'}\abs{U_{ij}}^2\abs{U_{ij'}}^2x_jx_{j'}\cos(\theta_j-\theta_{j'})\bigr)^{s/2}\bigr]^{1/s}}{(\sum_ix_i^r)^{1/r}}\\
    &=\sup_{x\in \mathbb{R}_+^d}\frac{\bigl[\sum_i\bigl(\sum_j \abs{U_{ij}}^2x_j\bigr)^s\bigr]^{1/s}}{(\sum_ix_i^r)^{1/r}}\\
    &=\sup_{x \in \mathbb{R}_+^d}\frac{\lVert\c2 x\rVert_s}{\lVert x\rVert_r}\,.
\end{split}  
\end{equation}
In going to the second line we have set $\phi_i = x_i e^{i\theta_i}$, and in going to the third line we have taken the supremum over the angles and rearranged the sums. Notice that in this step, it is important that all quantities appearing in the numerator are positive.
\end{proof}

\begin{lem}\label{lem:scaling}(Scaling symmetry).
The term in the r.h.s.\ of eq.\ \eqref{eq:normDef} is invariant under scaling $x\rightarrow k x$, for any $k>0$. Let $x_*$ be a value that maximizes the r.h.s., then so is $kx_*$. In particular, we can pick $k=\norm{x_*}_p^{-1}$. This means that in searching for the supremum, we can restrict to vectors with unit $p$-norm. This can often simplify calculations.
\end{lem}

In some simple examples, it is possible to calculate the norm of \c2 analytically. Here, we consider the two extremal cases in which the bases corresponding to $X$ and $Y$ are the same, or are mutually unbiased.

\begin{lem}\label{lem:mubNorm}(Mutually unbiased bases).
For mutually unbiased bases, the norm can be calculated analytically,
\begin{equation}
    \log\bigl\Vert\c2_{\text{\tiny MUB}}\bigr\Vert_{r\rightarrow s} =
        \Bigl(\frac{1}{s}-\frac{1}{r}\Bigr)\log{d}\,. 
\end{equation}
\end{lem}
\begin{proof}
In the mutually unbiased case, the matrix \c2 has constant entries, $(C^{(2)}_{\text{\tiny MUB}})_{ij}=1/d, \forall i,j$. 
Substituting this in the definition of the norm, we obtain
\begin{equation}
\begin{split}
    \bigl\Vert\c2_{\text{\tiny MUB}}\bigr\Vert_{r\rightarrow s} = d^{\frac1s-1}\sup_{x\in\mathbb{R}_+^d}\frac{\sum_i x_i}{(\sum_i x_i^r)^{1/r}}\,.
\end{split}
\end{equation}
To find the argsup, it is convenient to consider vectors with unit 1-norm, $\sum_ix_i=1$. This corresponds to $x$ be a valid probability distribution over $d$ elements, so let's call it $p$. We have
\begin{equation}
\begin{split}
    \bigl\Vert\c2_{\text{\tiny MUB}}\bigr\Vert_{r\rightarrow s} &= d^{\frac1s-1}\bigl(\sup_{p\in P_d}\frac{1}{\sum_i p_i^r}\bigr)^{1/r} \\
    &= d^{\frac1s-1}\bigl(\inf_{p\in P_d}\sum_i p_i^r\bigr)^{-1/r}\,,
\end{split}
\end{equation}
where we have used the fact that the $r$-root is monotonic increasing for $r>0$ and non-negative arguments. We can rewrite the expression on the right in terms of Rényi entropies,
\begin{equation}
    \bigl\Vert\c2_{\text{\tiny MUB}}\bigr\Vert_{r\rightarrow s} = d^{\frac1s-1}\bigl(\inf_{p\in P_d} \exp[(1-r)S_r(p)] \bigr)^{-1/r}\,,
\end{equation}
where 
\begin{equation}
    S_\alpha(p)\equiv\frac{1}{1-\alpha}\log\sum_ip_i^\alpha\,.
\end{equation}
It is known that Rényi entropies are maximized when the entries of $p$ are all equal to each other, and minimized when one of the entries is equal to 1. From this follows that 
\begin{equation}
    \log\bigl\Vert\c2_{\text{\tiny MUB}}\bigr\Vert_{r\rightarrow s} =\begin{dcases}
        \Bigl(\frac{1}{s}-\frac{1}{r}\Bigr)\log{d} & r\ge 1\,,\\
        \Bigl(\frac{1}{s}-1\Bigr)\log{d} & r<1\,.
    \end{dcases}
\end{equation}
\end{proof}
Notice that in our case $r=\lambda^{-1}$, which is always greater than 1, as $\lambda\in[0,1]$, so we are only interested in the first case above. Moreover, we have the following simple corollary. 
\begin{corr}
Consider a bipartite Hilbert space, $\mathcal{H}=\mathcal{H}_A\otimes \mathcal{H}_B$, then the norm for MUB is additive, 
\begin{equation}
    \log\bigl\Vert\c2_{\text{\tiny MUB, AB}}\bigr\Vert_{r\rightarrow s}=  \log\bigl\Vert\c2_{\text{\tiny MUB, A}}\bigr\Vert_{r\rightarrow s}+ \log\bigl\Vert\c2_{\text{\tiny MUB, B}}\bigr\Vert_{r\rightarrow s}\,.
\end{equation}
\end{corr}

\begin{lem}\label{lem:idNorm}(Norm for identity).
Next let's consider $\c2=\mathds{1}_d$, the identity matrix. This corresponds to $X=Y$. 
\begin{equation}\label{eq:normId}
    \log\norm{\mathds{1}_d}_{r\rightarrow s} =\begin{dcases}
        \Bigl(\frac{1}{s}-\frac{1}{r}\Bigr)\log{d} & r\ge s\,,\\
        0 & r<s\,.
    \end{dcases}
\end{equation}
\end{lem}
\begin{proof}
The norm is given by 
\begin{equation}
    \norm{\mathds{1}_d}_{r\rightarrow s}=\sup_{x\in \mathbb{R}_+^d}\frac{(\sum_i x_i^s)^{1/s}}{(\sum_j x_j^r)^{1/r}}\,.
\end{equation}
Let $f(x)$ be the quantity inside the sup, its derivative is given by
\begin{equation}
    \partial_if(x)=f(x)\Bigl(\frac{x_i^{s-1}}{\sum_kx_k^s}-\frac{x_i^{r-1}}{\sum_kx_k^r}\Bigr)\,.
\end{equation}
Critical points satisfy 
\begin{equation}
    x_i^{s-r}=\frac{\sum_kx_k^s}{\sum_jx_j^r}\,.
\end{equation}
Since the r.h.s.\ is independent of $i$, the only critical point (up to rescaling) is $\mathbf{1}=(1,\dots,1)$. The second derivative at this point is 
\begin{equation}
    \partial_i \partial_jf(x)\bigr\vert_{x=\mathbf{1}}=\frac{f(\mathbf{1})}{d}(r-s)\Bigl(\frac{1}{d}-\delta_{ij}\Bigr)\,.
\end{equation}
This matrix is easy to diagonalize: the spectrum has a zero eigenvalue corresponding to the scale symmetry, all the other eigenvalues are given by $(s-r)f(\mathbf{1})/d$. The critical point is a maximum as long as $r\ge s$;
otherwise, the maximum is obtained on the boundary of the domain, for $x_1=1$ and $x_{i\neq 1}=0$ (up to permutations). The final result is eq.\ \eqref{eq:normId}.
\end{proof}

\begin{lem}\label{lem:lowerBoundNorm}(MUB gives lower bound).
The norm of a generic $\c2$ is lower bounded by the norm of $\c2_{\text{\tiny MUB}}$ in the corresponding dimensions. 
\begin{equation}
    \bigl\Vert\c2\bigr\Vert_{r\rightarrow s}\ge \bigl\Vert\c2_{\text{\tiny MUB}}\bigr\Vert_{r\rightarrow s}\,.
\end{equation}
\end{lem}
\begin{proof}
This is easy to see
\begin{equation}
\begin{split}
    \bigl\Vert\c2\bigr\Vert_{r\rightarrow s}&=\sup_{x\in \mathbb{R}^d_+}\frac{\norm{\c2 x}_s}{\norm{x}_r}\ge \frac{\norm{\c2\mathbf{1}}_s}{\norm{\mathbf{1}}_r}\\
    &=d^{\frac{1}{s}-\frac1r} =\bigl\Vert\c2_{\text{\tiny MUB}}\bigr\Vert_{r\rightarrow s}\,.
\end{split}
\end{equation}
where $\mathbf{1}=(1, 1, \dots, 1)$.
\end{proof}

\begin{lem}\label{lem:upperBoundNorm}($\mathds{1}$ gives upper bound).
The norm of a generic $\c2$ is upper bounded by the norm of the identity, $\mathds{1}$, in the corresponding dimensions,
\begin{equation}
    \bigl\Vert\c2\bigr\Vert_{r\rightarrow s}\le \norm{\mathds{1}}_{r\rightarrow s}\,.
\end{equation}
\end{lem}
\begin{proof}
In general, \c2 is a convex combination of permutation matrices, $\c2 = \sum_\sigma p_\sigma \sigma$, where $\{\sigma\}$ is the set of all permutation matrices of a given dimension, and $\sum_\sigma p_\sigma=1$. We can lower bound the norm as follows,
\begin{equation}
\begin{split}
    \bigl\Vert\c2\bigr\Vert_{r\rightarrow s}&\le \sum_\sigma p_\sigma \norm{\sigma}_{r\rightarrow s}\\
    &=\norm{\mathds{1}}_{r\rightarrow s}\,.
\end{split}
\end{equation}
In the first line, we have used the triangle inequality, and in going to the last line, we have used that the norm is invariant under permutations, and that $\sum_\sigma p_\sigma=1$. 
\end{proof}
Putting together Lemmas \ref{lem:lowerBoundNorm} and \ref{lem:upperBoundNorm}), we find that the norm of generic matrices \c2 is upper bounded by the norm for MUB, and lower bounded by the norm of the identity. In equation,
\begin{equation}
    \bigl\Vert\c2_{\text{\tiny MUB}}\bigr\Vert_{\frac{1}{\mu}\rightarrow \frac{1}{1-\lambda}}\le \bigl\Vert\c2\bigr\Vert_{\frac{1}{\mu}\rightarrow \frac{1}{1-\lambda}}\le \norm{\mathds{1}}_{\frac{1}{\mu}\rightarrow \frac{1}{1-\lambda}}\,.
\end{equation}
Notice that the quantity entering eq.\ \eqref{eq:basicineq} is $-\log\norm*{\c2}$, so smaller norms give stronger bounds. The result above is then intuitive: measuring in MUB gives the strongest constraint, measuring in the same basis gives the weakest constraint, anything else falls in between. Moreover, since for $\mu+\lambda\le 1$ the upper and lower bounds coincide, we have that in this regime the norm of any matrix \c2 is given by the MUB result, eq.\ \eqref{eq:normMUB}. Unfortunately, in this regime, the log of the norm is positive and the bound in eq.\ \eqref{eq:basicineq} is useless.

\begin{lem}\label{lem:normExact}(Norm for $s\le r$).
For $s\le r$, and for any matrix $\c2$, the norm is given by
\begin{equation}
    \log\bigl\Vert\c2\bigr\Vert_{r\rightarrow s}=\Bigl(\frac{1}{s}-\frac{1}{r}\Bigr)\log{d}\,, \quad s\le r\,.
\end{equation}
This is the same result as in the MUB case.
\end{lem}
\begin{proof}
In Lemmas \ref{lem:lowerBoundNorm} and \ref{lem:upperBoundNorm}, we have showed that the MUB case also gives a lower bound that the norm of a generic \c2 falls in between the MUB case and the identity case. Moreover, in Lemma \ref{lem:idNorm}, we have showed that the norm of $\mathds{1}$, for $s\le r$, is given by the MUB result. Combining these facts, we have 
\begin{equation}
    \bigl\Vert\c2_{\text{\tiny MUB}}\bigr\Vert_{r\rightarrow s}\le \bigl\Vert\c2\bigr\Vert_{r\rightarrow s}\le \bigl\Vert\c2_{\text{\tiny MUB}}\bigr\Vert_{r\rightarrow s}\,,\quad r\ge s\,.
\end{equation}
From which follows that $\norm*{\c2}_{r\rightarrow s}=\norm*{\c2_{\text{\tiny MUB}}}_{r\rightarrow s}$.
\end{proof}
Incidentally, the regime $s\le r$ is also the one in which the efficient numerical algorithm of \cite{bhaskara2011approximating} can be used to calculate the norms. To the knowledge of the authors, there is no known efficient algorithm to calculate these norms away from this regime.

\begin{lem}
(KMU limit).
Consider the regime in which $\mu=\lambda\rightarrow 1$, such that $s\rightarrow \infty$ and $r\rightarrow 1$.
In this limit the norm is given by 
\begin{equation}
    \bigl\Vert\c2\bigr\Vert_{1\rightarrow \infty}= \max_{ij}\c2_{ij}\,.
\end{equation}
Since this is the same quantity appearing in the bound of Kraus, Masseen, and Uffing, we will call this the KMU limit.
\end{lem}

\begin{lem}
(Monotonicity of the norm).
The norm is monotonically decreasing as we increase $\mu$ or $\lambda$.
\end{lem}
\begin{proof}
This follows from the monotonicity property of the $p$-norm: $\norm{x}_p\ge \norm{x}_q$ if $p<q$ for $x\in \mathbb{R}^n$.
\end{proof}

\section{Evidences for conjecture \ref{conjecture}}\label{app:conjecture}
From Lemma \ref{lem:posVectors}, we know that to find the norm it is enough to optimize over vectors with positive entries, $x\in\mathbb{R}_+^d$. The norm is given by
\begin{equation}
    \bigl\Vert\c2\bigr\Vert_{r\rightarrow s}=\sup_{x\in \mathbb{R}^d_+}\frac{\bigl[\sum_k(\sum_l c_{kl}x_l)^s\bigr]^{1/s}}{\bigl[\sum_k x_k^r\bigr]^{1/r}}\,,
\end{equation}
where the matrix elements of \c2 satisfy $\sum_ic_{ij}=1$ and $\sum_jc_{ij}=1$.\footnote{Notice that this is true for \textit{any} doubly stochastic matrix. In our case, we have a little more: $c_{ij}=\abs{U_{ij}}^2$, for some unitary $U$. So far, we haven't managed to use this fact to simplify or extend the result above.} 
Let $f(x)$ be the quantity inside the sup. Our goal is to find the global maximum of this function.

The first piece of evidences for conjecture \ref{conjecture} is the following lemma.
\begin{lem}\label{lem:localMax}
(MUB regime in $d>2$).
The vector $\mathbf{1}\equiv(1,1,\dots, 1)$ is a local maximum of the function 
\begin{equation}
    f(x)=\frac{\bigl[\sum_k(\sum_l c_{kl}x_l)^s\bigr]^{1/s}}{\bigl[\sum_k x_k^r\bigr]^{1/r}}\,,
\end{equation}
with $\sum_ic_{ij}=1$ and $\sum_jc_{ij}=1$, as long as 
\begin{equation}\label{eq:condLocMax}
    \frac{1-\mu}{\mu}\frac{1-\lambda}{\lambda}\ge\sigma_2^2\,,
\end{equation}
where $\sigma_2$ is the second largest singular value of \c2.
\end{lem}
\begin{proof}
 The derivative of $f(x)$ is given by 
\begin{equation}
    \partial_i f(x) = f(x)\Bigl[\frac{\sum_k(\sum_l c_{kl}x_l)^{s-1}c_{ki}}{\sum_k(\sum_l c_{kl}x_l)^s}-\frac{x_i^{r-1}}{\sum_k x_k^r} \Bigr]\,.
\end{equation}
Finding all the zeros of this function is hard: if this was not the case, we could compute the norm easily. However, notice that we know one $x$ such that $\partial_i f(x)=0$. This is the point invariant under all permutations $ \mathbf{1}\equiv(1,1,\dots, 1)$. To check if this is at least a local maximum, we need to study the second derivative, which is given by
\begin{equation}
\begin{split}
    \partial_{i}\partial_j f(x) = \frac{\partial_{i} f(x)\partial_{j} f(x)}{f(x)}+ f(x)\Bigl[
        &(s-1)\frac{\sum_k(\sum_l c_{kl}x_l)^{s-2}c_{ki}c_{kj}}{\sum_k(\sum_l c_{kl}x_l)^s}\\
        &- s\frac{\sum_k(\sum_l c_{kl}x_l)^{s-1}c_{ki}\sum_{\tilde{k}}(\sum_{\tilde{l}} c_{\tilde{k}\tilde{l}}x_{\tilde{l}})^{s-1}c_{\tilde{k}j}}{[\sum_k(\sum_l c_{kl}x_l)^s]^2}\\
        &+r \frac{x_i^{r-1}x_j^{r-1}}{(\sum_k x_k^r)^2}
        -(r-1)\frac{\delta_{ij}x_i^{r-2}}{\sum_k x_k^r} \Bigr]\,.
\end{split}
\end{equation}
This expression considerably simplifies for $x=\mathbf{1}$,
\begin{equation}    
\partial_{i}\partial_j f(x)\Bigr\vert_{x=\mathbf{1}} = d^{\frac{1}{s}-\frac{1}{r}-1}\Bigl[\frac{r-s}{d}+(s-1)\sum_kc_{ki}c_{kj}-(r-1)\delta_{ij}\Bigr]\,,
\end{equation}
As long as all the eigenvalues of this matrix are negative, $x=\mathbf{1}$ is a local maximum. 

Setting $r=1/\mu$ and $s=1/(1-\lambda)$, and rearranging the terms we obtain 
\begin{equation}    
\partial_{i}\partial_j f(x)\Bigr\vert_{x=\mathbf{1}} = -\frac{d^{-\mu-\lambda}}{\mu(1-\lambda)}\bigl[(1-\mu)(1-\lambda)(\delta_{ij}-1)-\mu\lambda\bigl(\sum_kc_{ki}c_{kj}-1\bigr)\bigr]\,,
\end{equation}
Since the factor in front of the parenthesis is always negative, we might as well consider the Hamiltonian 
\begin{equation}    
H = (1-\mu)(1-\lambda)(\mathds{1}-M_1)-\mu\lambda(\c2_T\c2-M_1)\,,
\end{equation}
where $\c2_T$ denotes the transpose of \c2, and $M_1$ is the matrix with all entries equal to $1/d$. When the spectrum of this Hamiltonian is positive, $x=\mathbf{1}$ is a local maximum of $f(x)$.
 
For $\mu\lambda=0$, we are left with $\mathds{1}-M_1$, which has spectrum $(0, 1, 1, \dots, 1)$. The zero eigenvalue corresponds to $v=\mathbf{1}$. In fact, this is an eigenvector of $H$ with zero eigenvalue for any value of $\mu$ and $\lambda$: this flat direction corresponds to the freedom of rescaling $x$ by a constant $x\rightarrow x+\lambda x$, and plays no role. As we change $\mu$ and $\lambda$, the eigenvalues change: as soon as one of them crosses zero, $\mathbf{1}$ is not anymore a local maximum. To understand when this happens, it is convenient to work in the space orthogonal to $\mathbf{1}$, such that we don't need to worry about the zero eigenvalue. On this subspace the matrix reduces to 
\begin{equation}
    H' = (1-\mu)(1-\lambda)\mathds{1}-\mu\lambda\c2_T\c2 \,,
\end{equation}
because $M_1=\mathbf{1}\cdot \mathbf{1}^T$ is identically zero in the space orthogonal to $\mathbf{1}$.

This matrix $H'$ is diagonal in the basis in which $\c2_T\c2$ is diagonal, its eigenvalues are given by
\begin{equation}
    \lambda_k = (1-\mu)(1-\lambda)-\mu\lambda\sigma_k^2 \,, \quad k=2,\dots, d\,,
\end{equation}
where $\sigma_k^2$ are the eigenvalues of $\c2_T\c2$, i.e.\ the singular values of \c2, ordered such that $\sigma_1\ge \sigma_2 \ge \dots \ge \sigma_d$. The Birkhoff-von Neumann theorem implies that these are less or equal than 1, $\sigma_i\le 1$. The largest eigenvalue, $\sigma_1^2=1$ corresponds to $\mathbf{1}$, and it is not included in the subspace we are working in. This is why we only consider $k\ge 2$. Clearly the first $\lambda_k$ to become negative, is the second. So the spectrum of $H'$ and hence of $H$ is certainly positive as long as $\lambda_2\ge0$. This concludes our proof.
\end{proof}

Notice that since $\sigma_2^2\le 1$, the condition \eqref{eq:condLocMax} is certainly satisfied for $\mu+\lambda\le1$. In this regime we know from Lemma \ref{lem:normExact} that the norm of any \c2 is given by the MUB result. This implicitly implies that $\mathbf{1}$ is also a global maximum. For $\mu+\lambda>1$, we don't have a general argument why $\mathbf{1}$ should be also the global maximum. This is why in the main text, we have only a conjecture and not a theorem. 

In $d=2$, we can plot the function $f$ and numerically check that the conjecture is correct. The most general $2\times 2$ \c2 is given by \eqref{eq:2dDoublyStoc}. The norm is given by
\begin{equation}
    \bigl\Vert\c2\bigr\Vert_{\frac1\mu\rightarrow \frac{1}{1-\lambda}}=\sup_{x\in \mathbb{R}^2_+}\frac{[(x_0\cos^2\theta+x_1\sin^2\theta)^\frac{1}{1-\lambda}+(x_0\cos^2\theta+x_1\sin^2\theta)^\frac{1}{1-\lambda}]^{1-\lambda}}{[x_0^\frac{1}{\mu}+x_1^{\frac1\mu}]^\mu}\,.
\end{equation}
Since the argument of the sup is invariant under $x_0\leftrightarrow x_1$, we can assume w.l.o.g.\ that $x_0\neq 0$. Then we can rewrite the expression as 
\begin{equation}
    \bigl\Vert\c2\bigr\Vert_{\frac1\mu\rightarrow \frac{1}{1-\lambda}}=\cos^2\theta\cdot\sup_{z>0}\frac{[(1+z\tan^2\theta)^\frac{1}{1-\lambda}+(z+\tan^2\theta)^\frac{1}{1-\lambda}]^{1-\lambda}}{[1+z^{\frac1\mu}]^\mu}\,,
\end{equation}
where $z=x_1/x_0$. Let $f(z)$ be the argument of the sup. The symmetry under $x_0\leftrightarrow x_1$ becomes an inversion symmetry, i.e.\ $f(z)=f(1/z)$. Therefore, it is enough to consider $z\in[0,1]$. We want to check that $\mathbf{1}$, which corresponds to $z=1$, is a global maximum as long as \eqref{eq:condLocMax} is satisfied. We consider the case $\mu=\lambda$; the critical value at which $z=1$ ceases to be a local maximum is
\begin{equation}
    \mu_* = \frac{1}{2\cos^2{\theta}}\,.
\end{equation}
In Fig.\ \ref{fig:2dConj}, we plot $f(z)/f(1)$ for $\theta=\pi/8$ as we vary $\mu$ around $\mu_*$. We see that $z=1$ is a global maximum for $\mu\le\mu_*$. Crossed the critical value the function develops two new global maxima (the one you see in the figure, and its mirror under $z\rightarrow 1/z$). In the bottom plot, we zoom in on $\mu_*$. We observe that the maximum develops at $z=1$ and smoothly move away as we increase $\mu$. 
\begin{figure}
    \centering
\includegraphics[width=0.45\textwidth]{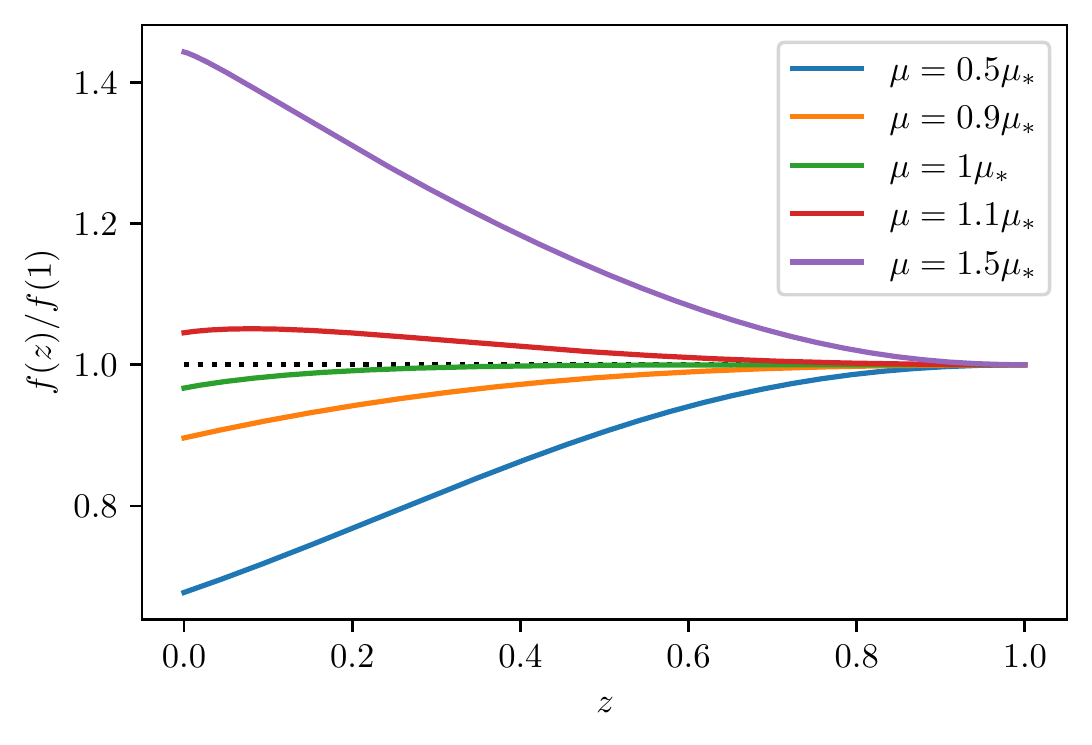}
    \includegraphics[width=0.45\textwidth]{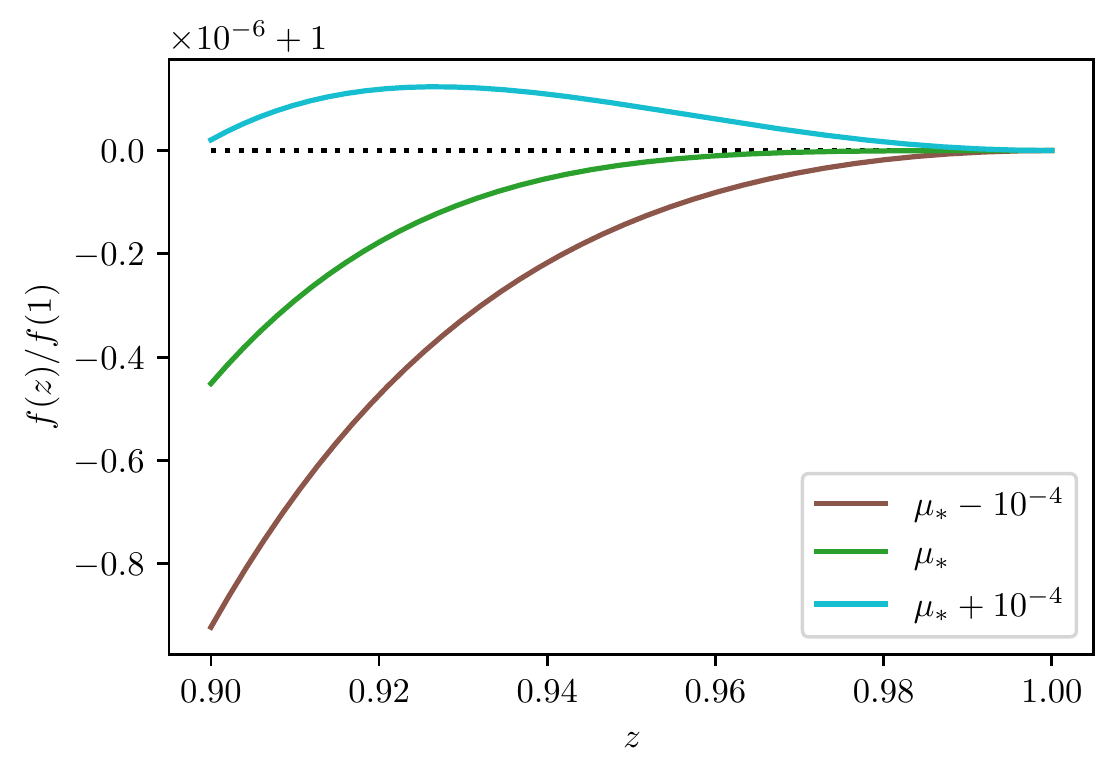}
    \caption{Plot of $f(z)/f(1)$ for \c2 given by \eqref{eq:2dDoublyStoc} and $\theta=\pi/8$ and for various values of $\mu$. In the bottom figure we zoom in on $\mu_*$.}
    \label{fig:2dConj}
\end{figure}
We suspect that a similar story holds in higher dimensions, but we have so far failed to prove this. The best we can say that numerically checking the entropic inequalities resulting from the conjecture we have not found any violation. 

\section{Additional calculations}\label{app:calculations}
We explain here some calculations that we have skipped for brevity in the main text.

\subsection{Proof of eq.\ \eqref{eq:randomConj}.}
Using \eqref{eq:normMUB} in eq.\ \eqref{eq:randomness}, we find
\begin{equation}
    H(X\vert E)\ge \max_{\mu, \lambda} \mu \Delta_Y -(1-\lambda)\Delta_X\,,
\end{equation}
where as in the main test $\Delta_X\equiv \log{d}-H(X)$, $\Delta_Y\equiv\log{d}-H(Y)$. Both $\Delta_X$ and $\Delta_Y$ are non-negative, so we would like to make $\mu$, $\lambda$ as big as possible. For MUB the expression above is always true and we can take $\mu=\lambda=1$, which leads to $H(X\vert E)\ge \log d-H(Y)$. For other measurements, conjecture \ref{conjecture} states that the MUB result is valid as long condition \eqref{eq:conjecture} is satisfied. This prevents us from increasing $\mu$, $\lambda$ indipendently. Since the allowed region for $\mu$, $\lambda$ is convex, we know that the maximum is obtained its boundary. Here, we can find $\mu$ as a function $\lambda$
\begin{equation}
    \mu_{\partial}=\frac{1-\lambda}{1-\lambda(1-\sigma_2^2)}\,.
\end{equation}

Plugging this in our bound for $H(X\vert E)$, we find 
\begin{equation}
     H(X\vert E)\ge \Delta_Y\max_{\lambda} \frac{1-\lambda}{1-\lambda(1-\sigma_2^2)}-(1-\lambda)\gamma^2\,,
\end{equation}
where $\gamma=\sqrt{\Delta_X/\Delta_Y}$. Wlog we can assume that $\gamma\le 1$, if this is not the case flip the roles of $X$ and $Y$. The function we want to maximize, which we denote with $f(\lambda)$, is positive for $\lambda\in[0,1]$, provided $\gamma\le1$. Moreover, we have $f(0)=1-\gamma^2$ and $f(1)=0$. The function has critical points at 
\begin{equation}
    \lambda_+=\frac{\gamma+\sigma_2}{\gamma(1-\sigma_2^2)}\,,\quad  \lambda_-=\frac{\gamma-\sigma_2}{\gamma(1-\sigma_2^2)}\,.
\end{equation}
For every value of $\gamma, \sigma_2$, $\lambda_+\ge 1$; while $\lambda_-\in[0,1]$ if $\gamma\ge\sigma_2$. We conclude that if $\gamma\ge\sigma_2$, the function is maximized by $\lambda=\lambda_-$; if $\gamma<\sigma_2$, the function is maximized at $\lambda=0$. Plugging these values in the expression, for $f(\lambda)$, we arrive to eq.\ \eqref{eq:randomConj}. Notice that for MUB, $\sigma_2=0$ and we recover the result we found at the beginning of this subsection. 

\subsection{Proof of \eqref{eq:entDet2}}
The proof follows \cite{alberto}. First, we use the fact that conditional Shannon entropies are concave, to show that for separable states, $\rho_{AB}=\sum_i p_i \rho_i^A\otimes \rho_i^B$, we have
\begin{equation}
    \lambda H(X_{AB})+\mu H(Y_{AB}) \ge \sum_i p_i (\lambda H(X_A)_{\rho_i^A}+\mu H(Y_A)_{\rho_i^A})+\lambda H(X_B)+\mu H(Y_B)\,.
\end{equation}
Then we apply eq.\ \eqref{eq:basicineq} to $A$ and $B$ separately, 
\begin{equation}
    \lambda H(X_{AB})+\mu H(Y_{AB}) \ge \sum_i p_iS(\rho_i^A)+S(\rho_B)-\log\bigl\Vert\c2_A\bigr\Vert_{\frac{1}{\mu}\rightarrow\frac{1}{1-\lambda}}-\log\bigl\Vert\c2_B\bigr\Vert_{\frac{1}{\mu}\rightarrow\frac{1}{1-\lambda}}\,.
\end{equation}
Finally one uses that $\sum_i p_iS(\rho_i^A)\ge0$ to get to eq.\ \eqref{eq:entDet2}. Notice that we can repeat the same proof flipping the roles of $A$ and $B$, so we can replace $S(\rho_B)$ with $\max(S(\rho_A), S(\rho_B))$. 

Notice that by directly applying \eqref{eq:basicineq} to $AB$, we obtain 
\begin{equation}
    \lambda H(X_{AB})+\mu H(Y_{AB}) \ge S(\rho_{AB})-\log\bigl\Vert\c2_{AB}\bigr\Vert_{\frac{1}{\mu}\rightarrow\frac{1}{1-\lambda}}\,.
\end{equation}
So the first bound can be violated by some entangled state only if 
\begin{equation}
    S(\rho_{AB})-\log\bigl\Vert\c2_{AB}\bigr\Vert_{\frac{1}{\mu}\rightarrow\frac{1}{1-\lambda}} \le S(\rho_B)-\log\bigl\Vert\c2_A\bigr\Vert_{\frac{1}{\mu}\rightarrow\frac{1}{1-\lambda}}-\log\bigl\Vert\c2_B\bigr\Vert_{\frac{1}{\mu}\rightarrow\frac{1}{1-\lambda}}\,,
\end{equation}
for some choice of measurements, and state $\rho_{AB}$. It is easy to show that $\log\norm*{\c2_{AB}} \ge \log\norm*{\c2_A}+\log\norm*{\c2_B}$, so it is sufficient for $S(\rho_{AB})\le S(\rho_B)\,.$ This is true, for example, for entangled pure states.

\subsection{Proof of eq.\ \eqref{eq:optimalMuLambda}.}
We want to maximize $\lambda \Delta_X+\mu \Delta_Y$ under the constraint \eqref{eq:conjecture} and $\mu,\lambda\in[0,1]$. The allowed region for $\mu, \lambda$ is convex, so the maximum will be on its boundary. Here we can find $\lambda$ as a function of $\mu$,
\begin{equation}
    \lambda_{\partial}=\frac{1-\mu}{1-\mu(1-\sigma_2^2)}\,.
\end{equation}
Let $f(\mu)=\lambda_\partial \Delta_X+\mu \Delta_Y$. The first derivative is given by
\begin{equation}
    f'(\mu)=\Delta_Y-\frac{\sigma_2^2}{[1-\mu(1-\sigma_2^2)]^2}\Delta_X\,.
\end{equation}
This is zero for 
\begin{equation}
    \mu_\pm = \frac{1\pm \sigma_2\gamma}{1-\sigma_2^2}\,.
\end{equation}

The second derivative is negative at $\mu_-$, and positive at $\mu_+$, for every value of $\gamma$. Therefore, the function $f(\mu)$ grows until $\mu_-$, decays from $\mu_-$ until $\mu_+$, and then grows again. Notice that $\mu_+\ge1$ for every value of $\gamma$. For $\gamma<\sigma_2$, also $\mu_-\ge1$ and the best we can do is to set $\mu=1$, for which $\lambda=0$. For $\sigma_2\le\gamma\le1/\sigma_2$, $\mu_-\in [0,1]$ and the maximum is given by $\mu_-$ itself. For $\gamma>1/\sigma_2$, $\mu_-<0$ and the best we can is to set $\mu=0$, for which $\lambda=1$. 

\subsection{Details about Fig.\ \ref{fig:werner}.}
To generate the figure, we take both $X_A$ and $X_B$ to be measurements in the computational basis. While we take $Y_A$ and $Y_B$ to be measurements in bases obtained with the following rotations
\begin{equation}\label{eq:orthogonalBasis}
    U_I = \begin{pmatrix}
        \cos\theta_I & \sin\theta_I\\
        \sin\theta_I & \cos\theta_I
    \end{pmatrix}\,, \quad I=A,B\,.
\end{equation}
We check whether \eqref{eq:entWitness}, for this class of operators, is able to detect entangled Werner states for various values of $\Phi<0$. The results are displayed in Fig.\ \ref{fig:werner}. The values of $\theta_A$ and $\theta_B$ for which entanglement is detected are shaded in green.
\end{document}